\documentclass{article}

\usepackage{amsmath, amsthm, amssymb}
\usepackage{graphicx}
\usepackage{graphics}
\usepackage{subfig}

\newtheorem{theorem}{Theorem}[section]

\newtheorem{proposition}[theorem]{Proposition}
\newtheorem{corollary}[theorem]{Corollary}

\begin{document}

\title{On nonsingular potentials of Cox-Thompson inversion scheme}

\author{Tam\'as P\'almai\footnote{Electronic mail:
palmai@phy.bme.hu} {\ }and Barnab\'as Apagyi\footnote{Electronic mail:
apagyi@phy.bme.hu} \\ \ \\
Department of Theoretical Physics\\
Budapest University of Technology and Economics\\
Budafoki ut 8., H-1111 Budapest, Hungary}

\maketitle

\begin{abstract}

We establish a condition for obtaining nonsingular potentials using the Cox–Thompson inverse scattering method with one phase shift. The anomalous singularities of the potentials are avoided by maintaining unique solutions of the underlying Regge–Newton integral equation for the transformation kernel. As a by-product, new inequality sequences of zeros of Bessel functions are discovered.

\ \\

\end{abstract}

\maketitle

\section{Introduction}
The Regge-Newton integral equation of the Cox-Thompson method
\cite{CT1} for the transformation kernel reads as
\begin{equation}\label{GL}
K(x,y)=g(x,y)-\int_{0}^{x} d t\,t^{-2}K(x,t)g(t,y),\qquad x\geq y,
\end{equation}
with the input symmetrical kernel defined as
\begin{equation}\label{CTang}
g(x,y)=\sum_{l\in S}\gamma_lu_l(x_<)v_l(x_>), \qquad \begin{array}{l}
x_{<}=\min(x,y),\\
x_{>}=\max(x,y).
\end{array}
\end{equation}
Here $u_l$ and $v_l$ means, respectively, the regular and
irregular Riccati-Bessel functions defined as
$u_l(x)=\sqrt{\frac{\pi x}{2}}J_{l+\frac{1}{2}}(x)$,
$v_l(x)=\sqrt{\frac{\pi x}{2}}Y_{l+\frac{1}{2}}(x)$,
and  the explicit expression holds for the
$\gamma_l$ numbers:
\begin{equation}
\gamma_l=\frac{\prod_{L\in T}[l(l+1)-L(L+1)]}{\prod_{l'\in S,
l'\neq l}[l(l+1)-l'(l'+1)]},\qquad l\in S,
\end{equation}
 with $|S|=|T|<\infty$, $T\subset (-1/2,\infty)$, $S\subset
(-1/2,\infty)$ and $S\cap T=\emptyset$ \cite{CT1,CT2}.

Let $\Omega$ denote the set of zeros of the determinant
\begin{equation}
D(x)=\det(C(x))
\end{equation}
with
\begin{equation}\label{determ}
[C(x)]_{lL}=\frac{u_L(x)v'_l(x)-u'_L(x)v_l(x)}{l(l+1)-L(L+1)}.
\end{equation}
In Ref. \cite{CT2} it is proved that equation (\ref{GL}) is uniquely
solvable for $x\in\mathbb{R}^+\setminus\Omega$ and the elements of
$\Omega$ are isolated points. Therefore the continuous solution of
equation (\ref{GL}) (if it exists) is unique.

In Ref. \cite{Chadan} it has been shown that $tq(t)$ is not integrable near
$\tilde{x}\in\Omega.$ Therefore the potential
$q(x):=-\frac{2}{x}\frac{d}{d x}\frac{K(x,x)}{x} $ corresponding
to the  Schr\"odinger equation has poles  of order (at least) 2 at
these isolated points  $\tilde x$. Such potentials are not in
$L_{1,1}(0,\infty)$ and we call them singular potentials.

To get non-singular potentials by the Cox-Thompson method is thus
in an intimate connection with the uniqueness of solution of
equation (\ref{GL}).
From now on the treatment is restricted to the one-term limit, i.e.,
to the case when $|S|=|T|=1$. Such a case represents a natural first
step compared to the uniqueness solution of the CT method with finite
number of phase shifts. Also, the one phase shift case is closely
related to the phenomenon of quantum resonance scattering (when the resonance-like enhancement of the
total cross section is mainly determined by a single partial wave) or
the Ramsauer-Townsend effect (when the electron-atom
interaction is governed mainly by the $p$-wave phase shift) \cite{Bransden}.

In the one-term limit, the numerator of
equation (\ref{determ}) becomes the Wronskian
\begin{equation}\label{Wronski}
W_{Ll}(x)=u_L(x)v'_l(x)-u'_L(x)v_l(x)=
\frac{\pi x}{2}\left(J_{L+\frac{1}{2}}(x)Y'_{l+\frac{1}{2}}(x)-J'_{L+\frac{1}{2}}(x)Y_{l+\frac{1}{2}}(x)\right).\phantom{00}
\end{equation}
To ensure a unique solution of the Regge-Newton integral equation (\ref{GL}), we shall establish  a
condition for $W_{Ll}(x)\ne 0, \quad x\in (0,\infty)$. This is
also the condition for constructing a non-singular potential
$q(x), \, x\in (0,\infty)$ at the one-term level $|S|=|T|=1$.

\section{Condition for constructing non-singular potentials from one specified phase shift}
Let $S=\{l\}$ and $T=\{L\}$, $L\neq l$. In order to get a
potential that belongs to the class $L_{1,1}(0,\infty)$ we shall
prove the next statement.
\begin{theorem}
$W_{Ll}(x)\ne 0,\quad x\in (0,\infty) \iff 0<|L-l|\le 1$.
\end{theorem}
\begin{proof}
First we prove that there exists $x>0$ such that $W_{Ll}(x)=0$ if
$|L-l|> 1$. Let $1+4k<l-L<3+4k$ with
$k\in\mathbb{ Z}$.
 Then the different sign of the Wronskian at the
origin $ W_{Ll}(x\to 0)=x^{L-l} \left[\frac{2^{l-L-1}(L+l+1)
\Gamma\left(l+\frac{1}{2}\right)}{\Gamma\left(L+\frac{3}{2}\right)}+O(x^{2l+1})\right]>0
$ and at the infinity $
W_{Ll}(x\to\infty)=\cos\left[\left(l-L\right) \frac{\pi}{2}\right]
< 0 $ clearly signals the existence of at least one zero position
$\tilde x$ for which $W_{Ll}(\tilde x)=0$ because of the
continuity of $W_{Ll}(x)$.

For the uncovered region of $3+4k<l-L<5+4k$ with $k\in\mathbb{
Z}\setminus \{-1\}$ we shall use the standard notation for the
$n$th zeros
$j_{L+\frac{1}{2},n},\,j'_{L+\frac{1}{2},n},\,y_{l+\frac{1}{2},n},\,y'_{l+\frac{1}{2},n}$
of the Bessel functions
$J_{L+\frac{1}{2}}(x),\,J_{L+\frac{1}{2}}'(x),\,Y_{l+\frac{1}{2}}(x),\,Y_{l+\frac{1}{2}}'(x)$.
Let now $l<L$. We  term {\it regular} sequence of zeros
 if the following interlacing holds for the $n$th and $(n+1)$th zeros:
$y_{l+\frac{1}{2},n}<j_{L+\frac{1}{2},n}<y_{l+\frac{1}{2},n+1}<j_{L+\frac{1}{2},n+1}. $ It is a simple matter to
see that the local extrema of $W_{Ll}(x)$ within the interval
$y_{l+\frac{1}{2},n}<x<j_{L+\frac{1}{2},n+1}$ possess the same sign in case of regular
sequence interlacing. This is because at the extremum positions
$y_{l+\frac{1}{2},n} $ and $j_{L+\frac{1}{2},n} $ of $W_{Ll}(x)$
the Wronskian simplifies to
\begin{equation}
W_{Ll}(x_{n})=\left\{\begin{array}{ll}
\frac{\pi x}{2}J_{L+\frac{1}{2}}(x_{n})Y'_{l+\frac{1}{2}}(x_{n})&\quad{\rm if}\quad Y_{l+\frac{1}{2}}(x_{n})=0,\,x_n=y_{l+\frac{1}{2},n}\\
-\frac{\pi x}{2}J'_{L+\frac{1}{2}}(x_{n})Y_{l+\frac{1}{2}}(x_{n})&\quad{\rm if}\quad
J_{L+\frac{1}{2}}(x_{n})=0,\,x_n=j_{L+\frac{1}{2},n}.\end{array} \right.
\end{equation}
Now, in case of any deviation from this regular sequence, e.g.,
when an {\it irregular} sequence
$y_{l+\frac{1}{2},n}<y_{l+\frac{1}{2},n+1}<j_{L+\frac{1}{2},n}$ is
first encountered at a particular $n=1,2,...$, one gets different
signs for the two consecutive extrema of the Wronskian at
$y_{l+\frac{1}{2},n}$ and $y_{l+\frac{1}{2},n+1}$, respectively.
This assumes the appearance of a zero position of $W_{Ll}(x)$
within the region $y_{l+\frac{1}{2},n}<x<y_{l+\frac{1}{2},n+1}$.
In summary, observing regular sequences of interlacing for all
$n>0$ is equivalent to the absence of roots of $W_{Ll}(x)$. To see
that in the considered region such deviation from the regular
sequence interlacing happens we present the following argument.
Let $L'<L$ such that $1+4k<l-L'<3+4k$. For $W_{L'l}$ the first
deviation from the regular sequence takes place at some $n'$. It
is easy to see that by increasing $L'$ to $L$ one cannot get a
regular sequence and the first deviation will occur at some $n\leq
n'$. Note that the case $L<l$ can be similarly treated.

Turning now to the most interesting domain of $0<|L-l|\le 1$, we
consider again the case $l<L$ and the regular sequence of zero
interlacing,
$y_{l+\frac{1}{2},n}<j_{L+\frac{1}{2},n}<y_{l+\frac{1}{2},n+1}<j_{L+\frac{1}{2},n+1}$.
As indicated above, its fulfillment ensures lack of root of the
Wronskian: $W_{Ll}(x)\ne 0, \, x \in (0<x<\infty)$. By noting that
any $n$th zero of a Bessel function is a strictly growing function
of the order it is sufficient to prove that
$y_{k,n}<j_{k+1,n}<y_{k,n+1}<j_{k+1,n+1}$, holds for $k\in
(0,\infty)$ and $n\in \mathbb{N}\setminus\{0\}$. The only unknown
inequality here is that of $j_{k+1,n}<y_{k,n+1}$. To prove its
validity we use the known intermediate relation
$j_{k,n+1}'<y_{k,n+1}$. Therefore, proving $j_{k+1,n}<j_{k,n+1}'$
will suffice. Consider the known relation
$J_{k+1}(j'_{k,n+1})=\frac{k}{j'_{k,n+1}}J_{k}(j'_{k,n+1})$ which
means that $J_{k+1}$ and $J_{k}$ have the same sign at
$x=j'_{k,n+1}$. Now because of the interlacing property
$j_{k,1}<j_{k+1,1}<j_{k,2}<...$ and the limit $J_k(x\to0)=0^+$
$\forall k>0$, this implies that the $n$th zero of $J_{k+1}(x)$
precedes the $(n+1)$th zero of $J'_{k}(x)$, i.e.
$j_{k+1,n}<j'_{k,n+1}<y_{k,n+1}$ which had to be proven. Note that
the case $L<l$ can be similarly treated.
\end{proof}

\begin{corollary}
In case of $|S|=1$, the Cox-Thompson inverse scattering scheme
yields a potential of the class $L_{1,1}(0,\infty)$ iff the
condition $0<|l-L|\leq1$ holds.
\end{corollary}

In the course of the proof we obtained the following result of its
own right:

\begin{proposition}
Denoting the $n$th root of the Bessel functions $J_\nu(x)$,
$Y_\nu(x)$, $J'_\nu(x)$, respectively, by $j_{\nu,n}$,
$y_{\nu,n}$, $j'_{\nu,n}$ then the following inequality is valid
for $\nu>0$: $j_{\nu+1,n}<j'_{\nu,n+1}$.
\end{proposition}
This proposition adds two new inequality sequences to the known
ones (see e.g. Ref. \cite{Abramowitz}):
$j_{\nu,n}<j_{\nu+1,n}<j'_{\nu,n+1}<j_{\nu,n+1},$ and
$j_{\nu+1,n}<y_{\nu,n+1}.$

\section{Construction of potentials from one phase shift}

One can construct a potential that possesses one specified phase
shift $\delta_l$ ($|S|=1$) by using the inversion scheme of Cox
and Thompson \cite{CT1,Apagyi}:
\begin{equation}
q(x)=-\frac{2}{x}\frac{d}{d x}\frac{K(x,x)}{x},
\end{equation}
\begin{equation}\label{Keq}
K(x,y)=\frac{l(l+1)-L(L+1)}{u_L(x)v'_l(x)-u'_L(x)v_l(x)}v_l(x)u_L(y),
\end{equation}
\begin{equation}\label{tan}
\tan(\delta_l-l\pi/2)=\tan(-L\pi/2).
\end{equation}
Relation (\ref{tan}) gives $L=l-\frac{2}{\pi}\delta_l+2n$, $n\in\mathbb{Z}$.
For $\delta_l\in [-\frac{\pi}{2},\frac{\pi}{2}]$
the Corollary  results in the choice of $n=0$. Therefore, for any
$\delta_l$, there is only one, easily identifiable non-singular
potential and an infinite number of singular potentials that the
Cox-Thompson method can produce.

For an example let us choose $l=0$ and $\delta_0=0$. In this  case
equation (\ref{tan}) yields $L=2n$, $n\in\mathbb{Z}$.
 $L=0$ ($n=0$) is not permitted by the assumption
$S\cap T=\emptyset$, however in order to get a non-singular
potential one may replace this $L=0$ by $L_n$ with
$\lim_{n\to\infty}L_n=0$. Using equation (\ref{Keq}) one gets at
$l=0$ and $L=L_n$
\begin{equation}
K_n(x,x)=\frac{-L_n(L_n+1)}{1+\varepsilon^1_n}(v_0(x)u_0(x)+\varepsilon^2_n).
\end{equation}
Since  $u_L(x)v'_l(x)-u'_L(x)v_l(x)$ and $v_l(x)u_L(x)$ are
continuous in $L$ and $u_l(x)v'_l(x)-u'_l(x)v_l(x)=1,\,\forall\,
l$, $\lim_{n\to\infty}\varepsilon^{1,2}_n=0$ holds. Thus
$\lim_{n\to\infty}q_n(x)\equiv0$ for $x>0$. This is the physical
solution. (See Fig. 1.)

\begin{figure}
\begin{center}
\includegraphics{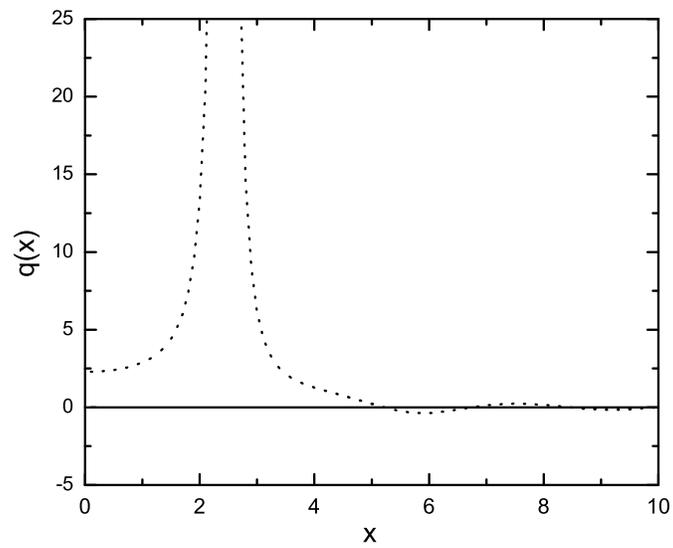}
\caption{Nonsingular (full line) and singular (dotted line)
potentials yielded by $\delta_{l=0}=0$ with $L\to 0$, and $L=2$.}
\end{center}\label{fig1}
\end{figure}

Now let $l=0$ and $L=2$. By the Corollary we cannot get an
integrable potential in this case because $|l-L|>1$ (see Fig.
1). In Ref. \cite{Ramm} it has been shown explicitly that equation
(\ref{GL}) is not uniquely solvable at some $x$  for this case.
However, while Ref. \cite{Ramm} suggests that this fact makes the
Cox-Thompson scheme useless, in this paper we have shown that in
order to get an integrable potential, the choice $L=2$ is not
permitted because the set $\Omega$ is not empty. On the other
hand, equation (\ref{tan}) and the Corollary provide a one-to-one
correspondence between the phase shift and the $L$ parameter of
the Cox-Thompson method at the one-term level. This correspondence
has the property that the potential constructed from $L$ belongs
to $L_{1,1}(0,\infty)$ and possesses the specified phase shift.

In Fig. 2 we add examples of the construction of an unique
potential in $L_{1,1}$ in the case of non-trivial phase shifts.
  Subfigure Fig. 2a shows the potentials for zero angular momentum
with $\delta_{l=0}=0.780$
 corresponding to $L=-0.497$ which is permitted by the Corollary (nonsingular case, full line)
 and $L=-0.497+2$ which violates the Corollary (one-singularity case, dotted line).
  Subfigure Fig. 2b shows the potentials for the $p$-wave phase shift $\delta_{l=1}=1.50$
 corresponding to $L=0.045$ [permitted by the Corollary (nonsingular case, full line)]
 and $L=0.044+4$  [violating the Corollary (two-singularity case, dotted line)].
We note that the second singularity of the dotted curve in Fig. 2b
lies out of the region shown and,
  as expected form the {\it Proof}, the singular potentials in Fig. 2a and 2b have one and two locally
  non-integrable region(s) corresponding to $1<|l-L|<3$ and $3<|l-L|<5$, respectively.

\begin{figure}[ht!]
\begin{center}
\subfloat[]{\includegraphics[scale=0.75]{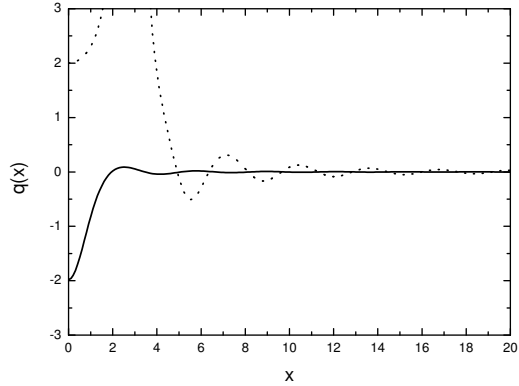}}\\
\subfloat[]{\includegraphics[scale=0.75]{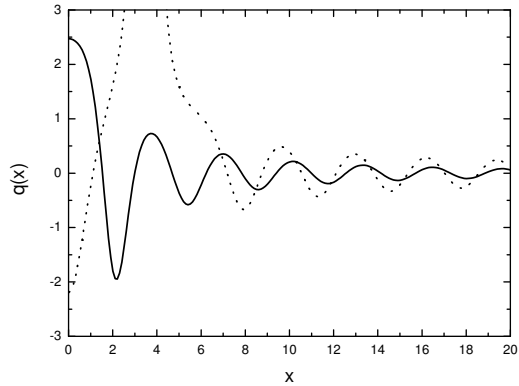}}
\caption{Nonsingular (full line) and singular (dotted line) potentials yielded by nontrivial phase shifts. (a)
$\delta_{l=0}=0.780$ with $L=-0.497$ (full line) and $L=-0.497+2$ (dotted line); (b)
$\delta_{l=1}=1.50$ with $L=0.045$ (full line) and $L=0.045+4$ (dotted line, see text).}
\end{center}\label{fig2}
\end{figure}

Applications (see Refs.
\cite{Apagyi,Melch,Apagyi2,Schumayer,Palm1,Palm2}) of the
Cox-Thompson scheme for $|S|>1$ suggest the existence of a
connection similar to the Corollary of Section 2 that specifies
one nonsingular potential out of the possible infinite singular
solutions. However such a theorem has, as yet, not been proven.

\section*{Acknowledgement}

The authors thank Professors Mikl\'os Horv\'ath and D\'aniel
Schumayer for reading the manuscript.


\begin{thebibliography}{12}

\bibitem[1]{CT1}
J. R. Cox and K. W. Thompson, J. Math. Phys. {\bf 11}, 805, (1970).

\bibitem[2]{CT2}
J. R. Cox and K. W. Thompson, J. Math. Phys. {\bf 11}, 815, (1970).

\bibitem[3]{Chadan}
K. Chadan and P. C. Sabatier, {\em Inverse Problems is Quantum
Scattering Theory} (Springer, New York, 1977), pp. 187-188.

\bibitem[4]{Bransden}
B. H. Bransden, {\em Atomic Collision Theory} (The Benjamin/Cummings Publishing Company, London, 1983).

\bibitem[5]{Abramowitz}
M. Abramowitz and I. A. Stegun, {\em Handbook of Mathematical
Functions} (Dover Publications, New York, 1972), pp. 360-371.

\bibitem[6]{Apagyi}
B. Apagyi, Z. Harman and W. Scheid, J. Phys. A: Math. Gen. {\bf
36}, 4815, (2003).

\bibitem[7]{Ramm}
A. G. Ramm, Applic. Anal. {\bf 81}, 833, (2002); A. G. Ramm, Mod.
Phys. Lett. B {\bf 22}, 2217, (2008).

\bibitem[8]{Melch} O. Melchert, W. Scheid and B. Apagyi, J. Phys. G {\bf32}, 849, (2006).

\bibitem[9]{Apagyi2} B. Apagyi W. Scheid, O. Melchert and D. Schumayer, Nuclear Physics A {\bf790}, 767c, (2007).

\bibitem[10]{Schumayer} D. Schumayer, O. Melchert, W. Scheid and B. Apagyi, J. Phys. B: At. Mol. Opt. Phys. {\bf41}, 035302, (2008).

\bibitem[11]{Palm1} T. P\'almai, M. Horv\'ath and B. Apagyi, J. Phys. A: Math. Theor. {\bf41},
235305, (2008).

\bibitem[12]{Palm2} T. P\'almai, M. Horv\'ath and B. Apagyi, Mod. Phys. Lett. B {\bf22},
2191, (2008).


\end{thebibliography}
\end{document}